\documentclass[letterpaper, 10pt, conference]{article}      

\usepackage{graphics} 
\usepackage{epsfig} 
\usepackage{supertabular}
\usepackage{graphicx}
\usepackage{mathtools}
\usepackage{times}
\usepackage{epstopdf}
\usepackage{csquotes}
\usepackage{color}
\usepackage{amssymb}
\usepackage{float}
\usepackage{pgfplots}
\pgfplotsset{compat=1.5}
\usepackage{subfig}
\usepackage{cite}
\newtheorem{theorem}{Theorem}[section]
\newtheorem{lemma}[theorem]{Lemma}
\newtheorem{proposition}[theorem]{Proposition}
\newtheorem{remark}[theorem]{Remark}
\newtheorem{proof}[theorem]{proof}

\usepackage{color}

\usepackage{kbordermatrix}
\def\be{\begin{equation}}
\def\ee{\end{equation}}
\def\bea{\begin{eqnarray}}
\def\eea{\end{eqnarray}}

\def\l{\label}

\newcommand{\ncom}{\newcommand}
\ncom{\beqn}{\begin{eqnarray*}} \ncom{\eeqn}{\end{eqnarray*}}
\ncom{\beq}{\begin{eqnarray}} \ncom{\eeq}{\end{eqnarray}}
\makeatletter
\title{Differential passivity like properties for a class of nonlinear systems}
\author{ K. C. Kosaraju$^1$, V. Chinde$^{2}$, R. Pasumarthy$^1$, A. Kelkar$^{2}$ and  N. M. Singh$^{3}$
\thanks{ $^1$Department
		of Electrical Engineering at IIT-Madras, Chennai, India {\tt\small ee13d015, ramkrishna@ee.iitm.ac.in}.}
	\thanks{$^2$Department of Mechanical Engineering at Iowa State University, Ames, USA {\tt\small vchinde@iastate.edu, akelkar@iastate.edu}.}
	\thanks{$^3$Department of Electrical Engineering at VJTI, Mumbai, India {\tt \small nmsingh59@gmail.com}.
	}%
}
\begin{document}
\maketitle
	\thispagestyle{empty}
	\pagestyle{empty}
	\begin{abstract}
	%
	%
	%
	In this paper we derive new passive maps akin to incremental passive maps, for a class of nonlinear systems using dynamic feedback and Krasovskii's method. Further using the passive maps we present a control methodology for stabilization to a desired operating point. This work is illustrated by designing a controller for a nonlinear building heating ventilating and air conditioning (HVAC) subsystem. 
	%
%
	\end{abstract}
	\section{Introduction}
	The second method of Lyapunov has been widely used for stability analysis of dynamical systems \cite{khalil1996noninear}. This method revolves around finding a suitable Lyapunov function that decreases along the system trajectories. Further, positive definite quadratic functions of state variables are usually a good candidate for Lyapunov functions. The classical Krasovskii's method \cite{krasovskiicertain} of generating Lyapunov functions also bears a similar form in terms of velocities (instead of states) and forms a candidate Lyapunov function for stability analysis.
	Apart from stability analysis, there has been a recent interest in incremental stability analysis \cite{angeli2002lyapunov} with targeted applications such as tracking, synchronization etc. Differential analysis is used for studying incremental stability properties through variational equations and has its roots in contraction theory \cite{lohmiller1998contraction,forni2014differential }. This analysis leads to a prolonged system and inherits new passivity properties \cite{forni2013differential} which extends the traditional passivity.


Passivity theory, with its roots in electrical network analysis, has been very useful in analyzing stability of a class of nonlinear systems \cite{l2gain}. Port-Hamiltonian systems are usually passive with respect to port variables that are power conjugates (eg: voltage and current, force and velocity). These natural port variables may not always help in achieving the desired stabilization criterion \cite{ortega2001putting}, in such cases one need to find alternate input-output passive maps \cite{jeltsema2004energy, venkatraman2010energy}.  Brayton Moser framework (BM)\cite{brayton1964theory, ortega2003power} is one such methodology that provides an alternative framework in providing these new passive maps. Contrary to the total energy as storage function for port-Hamiltonian systems, in Brayton Moser  framework the storage functions are derived from power. This resulted in passive maps with differentation on one of the port variables (eg: controlled voltages and the derivatives of currents, or the controlled currents and the derivatives of the voltages).

	Recently in \cite{ICCvdotidot}, the authors have shown that for systems in Brayton Moser framework, storage functions that are constructed using Krasovskii's Lyapunov functions 
	has yielded passive maps that has differentiation on both the port variables. Similar passive maps are obtained in \cite{fox2013population} to formulate stable games in input-output framework. To establish the result, the authors exploited the property that, dynamical systems in Brayton Moser formulations are contracting \cite{van2013differential,crouch1987variational}. This led to storage functions derived from Krasovskii's-type Lyapunov functions, which resulted in new passivity  property with ``differentiation at both the ports".
	Similar kind of studies have been carried out in order to extract new passivity properties of systems, namely differential passivity \cite{van2013differential} and incremental passivity \cite{l2gain}.
	In the case of incremental passivity the authors use the contraction property of the drift vector field to derive KYP like conditions for rending a system incrementally passive. Where as differential passivity allows one to verify the incremental passivity with a pointwise criterion. Later in the paper we detail the relations between the incremental and differential passivity properties with our new passive maps.
	In \cite{kosaraju2018stability}, the authors used tools and framework such as passivity, Krasovskii functions and BM framework to prove stability of continuous time primal-dual gradient descent equations of convex optimization problem. This framework draws its limitations for considering systems with constant input matrix. \\
	{\em Contribution:}
	In this paper, Krasovskii's method is used to derive sufficient conditions 
	for a class of non-linear systems via dynamic state feedback \cite{nijmeijer1990nonlinear}. The new passive maps obtained are used to shape the storage function for controller design. The proposed framework is demonstrated on a nonlinear HVAC subsystem namely thermal zone model.

 The organization of the paper is as follows. In Section
II, we discuss the Krasovskii formulation and the stability analysis of a class of nonlinear systems. In Section III, we demonstrate the proposed methodology on a building thermal zone model.
The results and discussion is provide in Section 
IV followed by conclusions presented in Section V.

{\em Notation:} If $f(x):\mathbb{R}^n\rightarrow \mathbb{R}$, then we represent $\dfrac{\partial f}{\partial x}=f_{x}$, $\dfrac{\partial^2 f}{\partial x^2}=f_{xx}$.
\section{Krasovskii formulation}
\subsection{Motivation}\label{sec::Topologically complete RLC circuits}

The dynamics of a Topologically complete RLC circuits \cite{jeltsema2003passivity} with regulated
voltage sources in series with inductors is described by
\beq
-L\frac{di}{dt}&=& \frac{\partial P}{\partial i}-B_sV_s\nonumber\\
C\frac{dv}{dt}&=& \frac{\partial P}{\partial v} \label{com_RLC}
\eeq
where $i$, $v$ denotes the current through the inductors $L$ and voltage across the capacitors $C$, $B_s\in \mathcal{R}^{n\times m}$ represents a constant input matrix and  the Mixed potential function $P(i,v)$ is given by
\beq \label{Mixpot}
P(i,v)&=&i^\top\Gamma v+G( i)-J( v)
\eeq
where $\Gamma\in \mathbb{R}^{n\times n}$ is skew-symmetric, $G(i)\geq 0$ and $J(v)\geq 0$ (Note that $L$ and $C$ are assume to be constant). Consider the following storage function
\beq\label{com_storage}
S(i,v)=\dfrac{1}{2}\left(\dfrac{di}{dt}\right)^\top L\dfrac{di}{dt}+\dfrac{1}{2}\left(\dfrac{dv}{dt}\right)^\top C\dfrac{dv}{dt}
\eeq
\begin{proposition} \label{prop_main}\cite{ICCvdotidot}
	Let $G_{ii}=\frac{\partial^2 }{\partial i^2}G, J_{vv}=\frac{\partial^2 }{\partial v^2}J$ be positive semi-definite then we have the following.
	 The system of equations \eqref{com_RLC} representing the dynamics of a complete RLC circuit in BM form, is passive with ports $B_s^\top\dfrac{di}{dt}$ and $\dfrac{dV_s}{dt}$.
\end{proposition}
\begin{proof}
Consider the storage function $S(i,v)$ defined in \eqref{com_storage}. The time derivative of $S(i,v)$ can be simplified as
	\beqn
	S_t &=& i_t^\top\left(-P_{ii}i_t-P_{vi}v_t+B_s\dfrac{dV_s}{dt}\right)\\&&+v_t^\top\left(P_{iv}i_t+P_{vv}v_t\right)\\
	&=& -i_t^\top P_{ii}i_t+v_t^\top P_{vv}v_t+i_t^\top B_s\dfrac{dV_s}{dt}
	\eeqn
	From \eqref{Mixpot} we get
	\beq\label{com_Sdot}
	\hspace{-3mm}\dfrac{d}{dt}S(i,v)\hspace{-3mm}&=&\hspace{-3mm}-i_t^\top G_{ii}i_t-v_t^\top J_{vv}v_t+i_t^\top B_s\dfrac{dV_s}{dt}
	\eeq
	From \eqref{com_RLC}, \eqref{com_storage}, \eqref{Mixpot} and \eqref{com_Sdot} it can be  proved that
	\beq
	\dfrac{d}{dt}S(i,v)&\leq&i_t^\top B_s\frac{dV_s}{dt}.
	\eeq
\end{proof}
	    \begin{remark}
Note the following in the proposition \ref{prop_main}:
\begin{itemize}
    \item[(i)]The nonlinear dynamical system given by, \eqref{com_RLC} with input $V_s=0$, is contracting with metric diag$\{L,C\}$ \cite{van2013differential,forni2013differential,crouch1987variational}.
    \item[(ii)] In deriving the result in proposition \ref{prop_main} we assumed that the input matrix $B$ as constant. The result is not obvious for a system with a state dependent input matrix $B$. That is the system represented by equations \eqref{com_RLC} is not passive with port variables $B_s(x)^\top\frac{di}{dt}$ and $\frac{dV_s}{dt}$.
\end{itemize}
	    \end{remark}
In this note, we present a methodology to derive new passive maps for systems with state dependent input matrix.
\subsection{A general nonlinear system}
Consider a nonlinear system of the form
\begin{equation}\label{gen_sys}
\dot{x}=f(x)+g(x)u
\end{equation}	
where $x\in \mathbb{R}^n$ is the state vector , $u\in \mathbb{R}^m$ ($m<n$) is the control input. $f(x):\mathbb{R}^n\rightarrow \mathbb{R}^n$ and $g(x):\mathbb{R}^n\rightarrow \mathbb{R}^m$, the input matrix are smooth functions.\\
{\em Assumption:}
\begin{itemize}
    \item[A1)]
    For a given $f(x)$ exist a symmetric positive definite matrix $M\in \mathbb{R}^{n \times n}$ satisfying
    \begin{equation}\label{A1}
        M\dfrac{\partial f}{\partial x}+\dfrac{\partial f}{\partial x}^\top M < 0
    \end{equation}
    This implies the dynamical system $\dot{x}=f(x)$ is contracting.
    \item[A2)]The full-rank left annihilator of input matrix also left annihilates its Jacobian. Let $g^{\perp}$ denotes left annihilator of the input matrix $g(x)$ i.e., $ g^{\perp}g$=0 then 
    \begin{equation}\label{A2}
        g^{\perp}\dfrac{\partial g}{\partial x}=0
      \end{equation}
    \item[A3)]$Mg(x)$ is Integrable.
\end{itemize}
\begin{proposition}\cite{lohmiller1998contraction}\label{prop::Contrating}
Consider system \eqref{gen_sys} with input $u=0$ satisfying Assumption A1. Then the resulting dynamical system is contracting.
\end{proposition}
	\begin{proof}
	Consider the Krasovskii Lyapunov function
	\begin{equation}\label{kras_lypunov}
	    V(x,\dot{x})=\dfrac{1}{2}\dot{x}^\top M \dot{x}.
	\end{equation}
	Then the time derivative of \eqref{kras_lypunov} along the trajectories of \eqref{gen_sys} with $u=0$ is
	\begin{align*}
	\small
	    \dfrac{d}{dt}V &= \dot{x}^\top M\ddot{x}= \dot{x}^\top M\left(\dfrac{\partial f}{\partial x}\dot{x}\right)\\&= \dot{x}^{\top}\left(M\dfrac{\partial f}{\partial x}+\dfrac{\partial f}{\partial x}^\top M\right)\dot{x} \leq 0
	\end{align*}
	    This implies the dynamical system $\dot{x}=f(x)$ is contracting in $\mathbb{R}^n$ with respect to the metric $M$ \cite{lohmiller1998contraction}.
	    \end{proof}
	    \begin{remark}
	    In assumption A1, one can consider a state dependent Riemannian metric $M(x)$, and replace equation \eqref{A1} with
	 \begin{eqnarray*}
	 M\dfrac{\partial f}{\partial x}+\dfrac{\partial f}{\partial x}^\top M +\dot{M}< 0.
	 \end{eqnarray*}
	    \end{remark}
	Denote
	\begin{equation}\label{alpha_1}
	\alpha=-\left(g^\top g\right)^{-1}g^\top \dot{g}.
	\end{equation}
	The following lemma will be instrumental in formulating our result.
 	\begin{lemma}\label{prop::alpha}
 	Consider an input matrix $g(x)$ satisfying assumption A2. Then
	\begin{equation}\label{stab_cond}
	    	\dot{g}+g \alpha=0
	\end{equation}
	if and only if $\alpha$ satisfies \eqref{alpha_1}.
	\end{lemma}
	\begin{proof}
	The {\em only if} part of the proof: consider the following full rank matrix $\begin{bmatrix}g^{\perp}\\g^\top \end{bmatrix}$. Now left multiplying $\left(\dot{g}+g \alpha\right)$ in \eqref{stab_cond} by $\begin{bmatrix}g^{\perp}\\g^\top \end{bmatrix}$ yields
	\begin{eqnarray*}
	\begin{bmatrix}g^{\perp}\\g^\top \end{bmatrix}\left(\dot{g}+g \alpha\right)&=&\begin{bmatrix}g^{\perp}\left(\dot{g}+g \alpha\right)\\g^\top \left(\dot{g}+g \alpha\right)\end{bmatrix}\\
	&=&\begin{bmatrix}g^{\perp}\dot{g}\\g^\top \left(\dot{g}-g \left(g^\top g\right)^{-1}g^\top \dot{g}\right)\end{bmatrix}\\
	&=&\begin{bmatrix}g^{\perp}\dfrac{\partial g}{\partial x}\dot{x}\\ \left(g^\top\dot{g}-g^\top g \left(g^\top g\right)^{-1}g^\top \dot{g}\right)	\end{bmatrix}\\
	&=&\begin{bmatrix}0\\ \left(g^\top\dot{g}-g^\top \dot{g}\right)	\end{bmatrix}\\
	&=&0
	\end{eqnarray*}
	By construction $\begin{bmatrix}g^{\perp}\\g^\top \end{bmatrix}$ is full rank matrix, hence $\dot{g}+g \alpha=0$. The {\em if} part of the proof
$$
	\dot{g}+g\alpha=0 \implies g^\top g \alpha=g^\top \dot{g}\implies \alpha=-(g^\top g)^\top g^\top \dot{g}.
$$
hence
\begin{eqnarray*}
\alpha=-(g^\top g)^\top g^\top \dot{g}\iff \dot{g}+g \alpha=0.
\end{eqnarray*}
	\end{proof}
\begin{figure}
	\centering
	\includegraphics[width=1\linewidth]{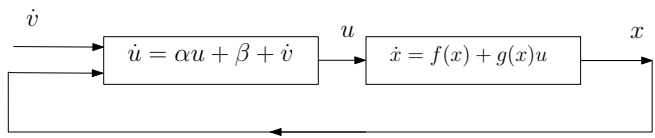}
	\caption{Interconnection of  dynamic state feedback \eqref{input_dyn}  to system \eqref{gen_sys}.}
	\label{fig:dyn_state_feedback}
\end{figure}
	Consider the following dynamic state feedback \cite{nijmeijer1990nonlinear} for system \eqref{gen_sys} (see Fig. \ref{fig:dyn_state_feedback})
	\begin{equation}\label{input_dyn}
	    \dot{u}= \alpha u+\beta +\dot{v}
	\end{equation}
	with $\alpha$ defined as in lemma \ref{prop::alpha}, $\beta=-g^\top M \dot{x}$ and $v\in \mathbb{R}^m$. The use of $\dot{v}$ in \eqref{input_dyn} rather than $v$ as new port variable will evident in the later part of the note. We have following theorem.
	\begin{theorem}\label{thm::main}
	Let the assumptions A1, A2 are satisfied. Then the system \eqref{gen_sys} together with \eqref{input_dyn} are passive with input $\dot{v}$ and output $y=g^\top M \dot{x}$.
	\end{theorem}
	\begin{proof}
	Consider  storage function of the form \eqref{kras_lypunov}. The time derivative of \eqref{kras_lypunov} along the trajectories of \eqref{gen_sys} and \eqref{input_dyn} is
	\begin{eqnarray*}
	\dfrac{d}{dt}V&=& \dot{x}^\top M\ddot{x}\\
	    &=& \dot{x}^\top M\left(\dfrac{\partial f}{\partial x}\dot{x}+\dot{g}u+g\dot{u}\right)\\
	    &=& \dot{x}^\top M\left(\dfrac{\partial f}{\partial x}\dot{x}+\dot{g}u+g\left(\alpha u+\beta +\dot v\right)\right)\\
	    &=& \dot{x}^{\top}\left(M\dfrac{\partial f}{\partial x}+\dfrac{\partial f}{\partial x}^\top M\right)\dot{x}\\&&+\dot{x}^\top M\left(\left(\dot{g}+g \alpha)\right)u+g\beta+g\dot v \right)\\
	    &\leq & \dot{v}^\top y
	\end{eqnarray*}
	where $y=g^\top M \dot{x}$ is also referred to as power shaping output. In step 1 and 2 we use system dynamics \eqref{gen_sys} and controller dynamics \eqref{input_dyn} respectively. In step 4 and 5 we used Proposition \ref{prop::Contrating} and lemma \ref{prop::alpha} respectively.
	\end{proof}
	\subsection{Control}
	The new passive maps obtained with differentiation at the port variables are further used for shaping the storage function. The controller is obtained are a result of the stability analysis treatment of the storage function.\\
	{\em Control objective:} To stabilize the system \eqref{gen_sys} at an non-trivial operating point $(x^{\ast},u^{\ast})$ satisfying
	\begin{equation}\label{Control_objective}
	    f(x^\ast)+g(x^\ast)u^\ast =0
	\end{equation}
	\begin{lemma}\label{prop::output_integrability}
	The output $y=g^\top M \dot{x}$ given in Theorem \eqref{thm::main} is integrable.
	\end{lemma}
	\begin{proof}
	From Assumption A3, we have that the function $Mg(x)$ is integrable, Poincare's Lemma ensures the existance of a function $\Gamma(x):\mathbb{R}^n\rightarrow \mathbb{R}^n$ such that
	\begin{equation}
	\dot{\Gamma}=(MG)^\top \dot{x}.
	\end{equation}
	\end{proof}
	By exploiting the integrability property of the output, the authors in \cite{chinde2016building}, have presented a methodology to construct the closed loop storage function whose minimum is at the desired operating point.
	Consider the storage function of the form
	\begin{equation}\label{clP_str}
	    V_d(x)=\dfrac{1}{2}k_1\dot{x}^\top M \dot{x}+\dfrac{1}{2}||\Gamma(x)-\Gamma(x^\ast)||^2_{k_i}.
	\end{equation}
	\begin{figure}
	\centering
	\includegraphics[width=1\linewidth]{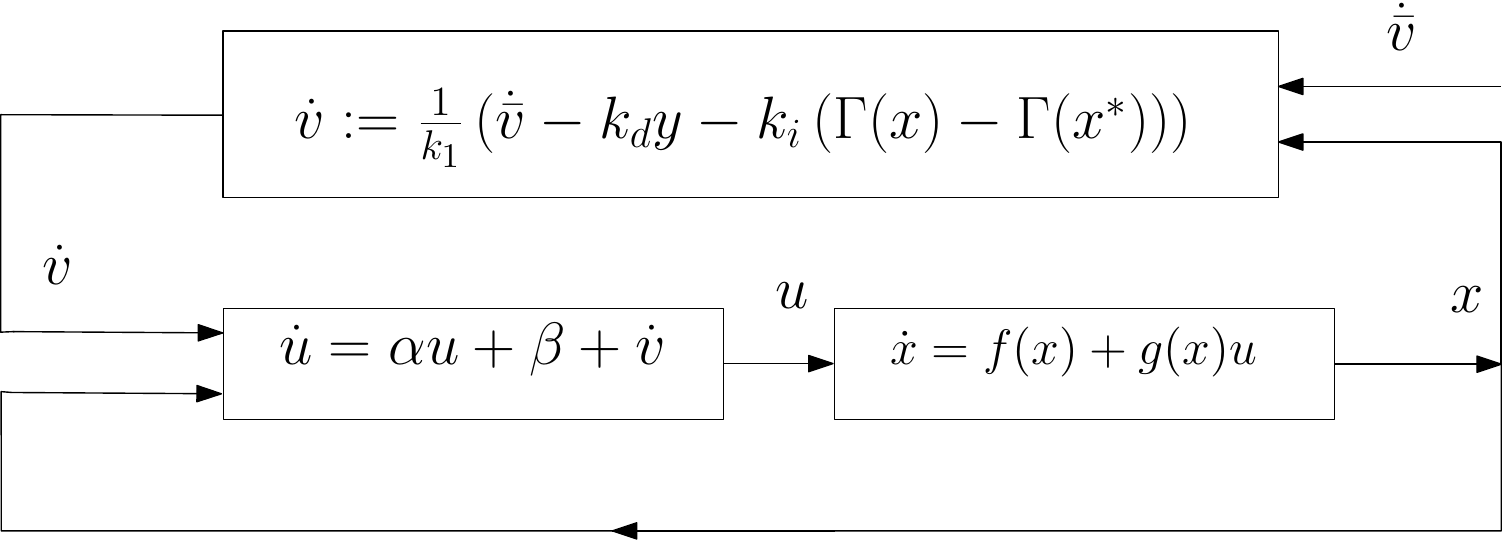}
	\caption{Interconnecting the controller \eqref{input_v_dyn} to dynamic state feedback system in Fig. \ref{fig:dyn_state_feedback}.}
	\label{fig:dyn_state_feedback_controller}
\end{figure}
	\begin{proposition}\label{porp::control}
	Consider system \eqref{gen_sys} together with \eqref{input_dyn} satisfying assumptions A1, A2 and A3. We define the mapping $v:\mathbb{R}^n\rightarrow\mathbb{R}^m$
	\begin{equation}\label{input_v_dyn}
	    \dot v:=\dfrac{1}{k_1}\left(\dot{\bar{v}}-k_d y-k_i \left(\Gamma(x)-\Gamma(x^{\ast})\right)\right)
	\end{equation}
	where $y=g^\top M \dot{x}$. Then the system of equation \eqref{gen_sys} and \eqref{input_dyn} are passive with port variables $\dot{\bar{v}}$ and $y$ (see Fig. \ref{fig:dyn_state_feedback_controller}). Further for $\dot{\bar{v}}=0$, the system is stable and $x^\ast$  as the stable equilibrium point. Furthermore if $y=0 \implies \lim_{t\rightarrow \infty}x(t)\rightarrow x^{\ast}$, then $x^\ast$ is asymptotically stable.
	\end{proposition}
	\begin{proof}
	The time derivative of the closed loop storage function \eqref{clP_str} is
	\begin{eqnarray*}
	\dfrac{d}{dt}V_d&=& k_1\dot{V}+y^\top k_i(\Gamma(x)-\Gamma(x^\ast))\\
	&\leq & y^\top\left(k_1\dot{v}+ k_i(\Gamma(x)-\Gamma(x^\ast))\right)\\
	&\leq & y^\top \dot{\bar{v}}
	\end{eqnarray*}
		This proves that the closed loop system is passive with storage function $V_d$, input $\dot{\bar{v}}$ and output $y$. Further for $\dot{\bar{v}}=0$ we have 
	\begin{equation*}
	    \dot{V}_d\leq -k_d y^\top y
	\end{equation*}
	and at equilibrium $x=x^\ast$ we have $\dot{v}=0$, further using this in \eqref{input_dyn} we can show that $\dot u=0$. This implies $(x^\ast,y^\ast)$ satisfy the control objective \eqref{Control_objective}, further concluding that system \eqref{gen_sys} is asymptotically stable with Lyapunov function $V_d$ and $x^\ast$ as the equilibrium point.
	\end{proof}
	\begin{remark}
Note the following.
		\begin{itemize}
	    \item[(1)] At the desired operating point one can show that  $\dot{u}-\alpha u-\beta=0$.  Hence, we have considered $\dot{u}= \alpha u+\beta +\dot{v} $, instead of $\dot{u}= \alpha u+\beta +v$ in equation \eqref{input_dyn}.
	    \item[(2)] Systems that are contracting always forget their initial conditions. That is, their final behaviour is always independent of the initial conditions. Hence, one need not worry about the initial conditions of the control input $u$ while implementing the control law \eqref{input_dyn} together with \eqref{input_v_dyn}.
	\end{itemize}
	\end{remark}
\section{Illustrative example: Temperature regulation of a building thermal zone}\label{buildingzone}
Thermal zone is an important component of heating ventilating and air conditioning (HVAC) subsystem. Although, there are different zone modeling strategies, for control purpose, lumped parameter models are commonly used \cite{ma2012predictive}. Lumped parameter models have resistance-capacitance (RC) interconnected network which represents interaction between zones and between zone and ambient. The capacitances represent the total thermal capacity of the wall, zone, and the resistances are
used to represent the total resistance that the wall offers to
the flow of heat from one side to other. To illustrate the proposed approach, we consider a simple two-zone case separated by a wall, where the surface is modeled as a 3R2C \cite{deng2010building} network as shown in Fig. \ref{fig:zonemodel}.
\begin{figure}[h!]
	\centering
	\includegraphics[width=1\linewidth]{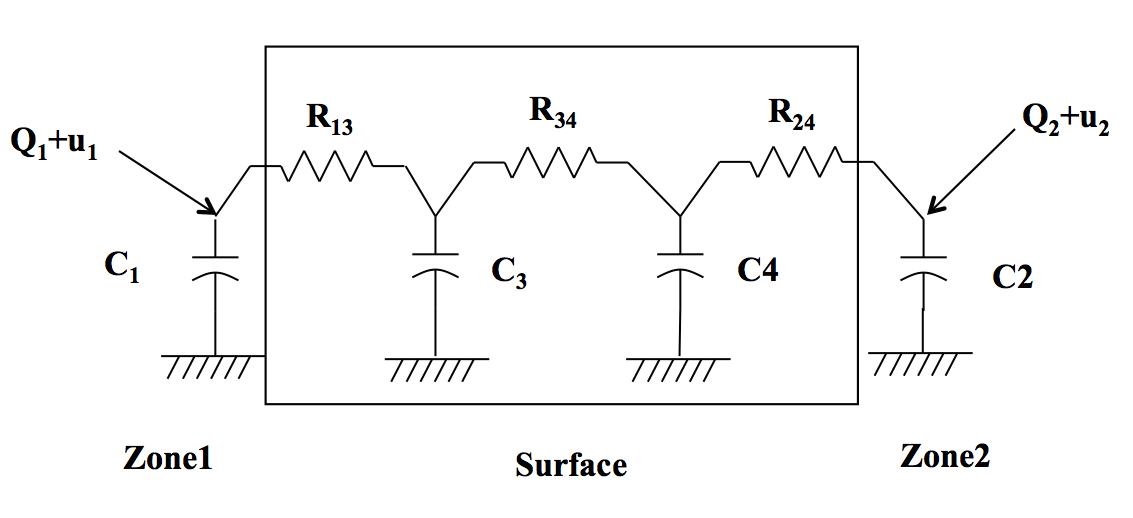}
	\caption{Lumped RC network model: Two zone case}
	\label{fig:zonemodel}
\end{figure}
	The nonlinear thermal model for the two zone case is given by \cite{chinde2016building}
	\begin{eqnarray}\label{dyn}
	C_1\dot{T_1}&=&\dfrac{T_3-T_1}{R_{31}}+\dfrac{(T_{\infty}-T_1)}{R_{10}}+u_1c_p(T_s-T_1)\nonumber\\		C_2\dot{T_2}&=&\dfrac{T_4-T_2}{R_{42}}+\dfrac{(T_{\infty}-T_2)}{R_{10}}+u_2c_p(T_s-T_2)\nonumber\\			
	C_3\dot{T_3}&=&\dfrac{T_1-T_3}{R_{13}}+\dfrac{(T_{4}-T_3)}{R_{34}}\\	C_4\dot{T_4}&=&\dfrac{T_2-T_4}{R_{42}}+\dfrac{(T_{3}-T_4)}{R_{34}}\nonumber
	\end{eqnarray}
In the above model, the inputs $u_1$ and $u_2$ denotes the mass flow rates. $T_{\infty}$, $T_s$ are ambient and supply air temperatures. Note that the inputs are coupled with the state (Temperatures $T_{1}$,$T_{2}$).
Denote the following:
\begin{eqnarray}\label{alphabeta}
\alpha=\begin{bmatrix}
	\frac{\dot{T}_1}{(T_s-T_1)} &0\\0&\frac{\dot{T}_2}{(T_s-T_2)}
\end{bmatrix},\text{and}\;\; \beta=\begin{bmatrix}
c_p(T_1-T_s)\dot{T}_1\\c_p(T_2-T_s)\dot{T}_2
\end{bmatrix}.
\end{eqnarray}
\begin{proposition}
The systems of equations \eqref{dyn}, and \eqref{input_dyn} with $\alpha$ and $\beta$ defined as in \eqref{alphabeta}, are passive with port variables $\dot{v}$ and $y$. where
\begin{eqnarray}
y(T)&=& c_p\begin{bmatrix}
\left(T_s-T_1\right)\dot{T}_1\\ \left(T_s-T_1\right)\dot{T}_2
\end{bmatrix}.
\end{eqnarray}
    
\end{proposition}
\begin{proof}
Let $C=\text{diag}\left\{C_1,C_2,C_3,C_4\right\}$.
 One can prove that the system \eqref{dyn} satisfies assumption (A1) given in equation \eqref{A1} by choosing $M=\text{diag}\left\{C_1,C_2,C_3,C_4\right\}$. \\
 The input matrix of \eqref{dyn} is $g(T)=[g_1(T),g_2(T)]$, where 
\begin{eqnarray*}
g_1(T)&=&\begin{bmatrix}
\dfrac{c_p}{C_1}(T_s-T_1)&0&0&0
\end{bmatrix}^\top,\\
g_2(T)&=&\begin{bmatrix}
0&\dfrac{c_p}{C_2}(T_s-T_2)&0&0
\end{bmatrix}^\top.
\end{eqnarray*}
Using left annihilator of $g(T)$, that is 
$$g^{\perp}(T)=\begin{bmatrix}0&0&1&0\\ 0&0&0&1\end{bmatrix}$$
one can show that
\begin{eqnarray}\label{inpu_matrix}
\begin{matrix}
g^{\perp}\dfrac{\partial g_1 }{\partial T}=0 &
g^{\perp}\dfrac{\partial g_2 }{\partial T}=0
\end{matrix}
\end{eqnarray}
Hence the input matrix $g(T)$ satisfies assumption A2. Now, we can use Proposition \eqref{prop::alpha} and show that $\alpha$ takes the same form, given in \eqref{alphabeta}. Finally from Theorem \ref{thm::main}, using
 \begin{eqnarray}\label{Storage_fun_1}
V(T)&=&\dfrac{1}{2}\dot{T}^\top M\dot{T}\\
&=&\dfrac{1}{2}\left(C_1\dot{T}_1^2+C_2\dot{T}_2^2+C_3\dot{T}_3^2+C_4\dot{T}_4^2\right)\nonumber
\end{eqnarray}
as storage function, the system of equations \eqref{dyn}, together with input dynamics \eqref{input_dyn} given by
\begin{eqnarray}\label{udot}
\begin{matrix}
\dot{u}_1=\left(\dfrac{u_1}{(T_s-T_1)}-c_p(T_s-T_1)\right)\dot{T}_1+\dot{v}_1\\
\dot{u}_2=\left(\dfrac{u_2}{(T_s-T_2)}-c_p(T_s-T_1)\right)\dot{T}_1+\dot{v}_2
\end{matrix}
\end{eqnarray}
are passive with port variables $\dot{v}$ and $y$.
\end{proof}
Now we can consider $v=[v_1,v_2]^\top$ as input for the combined equations \eqref{dyn}, \eqref{udot} and provide a control strategy using Proposition \eqref{porp::control}. Consider $a_1=(T_1^\ast-T_s)^2$, $a_2=(T_2^\ast-T_s)^2$, $k_d\geq0$ and $k_i>0$.
\begin{proposition}
The state feedback controller
\begin{eqnarray}\label{Control_port}
\begin{matrix}
\dot{v}_1\hspace{-3mm}&=&\hspace{-3mm}-k_dc_p\left(T_s-T_1\right)\dot{T}_1+\dfrac{1}{2}k_ic_p\left(\left(T_s-T_1\right)^2-a_1\right)\\ 
\dot{v}_2\hspace{-3mm}&=&\hspace{-3mm}-k_dc_p\left(T_s-T_1\right)\dot{T}_2+\dfrac{1}{2}k_ic_p\left(\left(T_s-T_2\right)^2-a_2\right)
\end{matrix}
\end{eqnarray}
asymptotically stabilizes the system of equations \eqref{dyn} and \eqref{udot} to the operating point $(T^\ast,u^\ast)$ satisfying \eqref{Control_objective}.
\end{proposition}
\begin{proof}
With $M=\text{diag}\{C_1,C_2,C_3,C_4\}$ and input matrix $g(T)$ in \eqref{inpu_matrix}, one can verify assumption A3. Hence from Proposition  \ref{prop::output_integrability}, we can show that 
\begin{eqnarray}\label{gamma}
\Gamma(T)=-\dfrac{1}{2}c_p\begin{bmatrix}(T_1-T_s)^2\\(T_2-T_s)^2\end{bmatrix} 
\end{eqnarray}
satisfies $\dot{\Gamma}(T)=y(T)$. Further proof directly follows from Proposition \ref{porp::control} using $\Gamma(T)$ in \eqref{gamma}. It can also be proved by taking the time derivative of Lyapunov function \eqref{clP_str} along the trajectories of \eqref{dyn} and \eqref{udot} as shown below
\begin{eqnarray*}
\dot{V}_d&=&k_1\dot{T}^\top M\ddot{T}+k_i(\Gamma(T)-a)^\top \dot{\Gamma}(T)\\
&=& -\dfrac{k_1}{R_{13}}\left(\dot{T}_1-\dot{T}_3\right)^2-\dfrac{k_1}{R_{24}}\left(\dot{T}_2-\dot{T}_4\right)^2\\&&-\dfrac{k_1}{R_{34}}\left(\dot{T}_3-\dot{T}_4\right)^2-\dfrac{k_1}{R_{10}}\left(\dot{T}_1^2+\dot{T}_2^2\right)\\
&&+\dot{T}^\top M\dfrac{d}{dt}\left(g(T)u\right)+k_i(\Gamma(T)-a)^\top y(T)\\
&\leq &\dot{T}^\top\left(\dot{g}u+g\dot{u}\right)+k_i(\Gamma-a)^\top y\\
&= &\dot{T}^\top M\left(\dot{g}u+g(\alpha u+\beta+v)\right)+k_i(\Gamma-a)^\top y\\
&\leq &\dot{T}^\top M\left((\dot{g}+g\alpha)u+gv\right)+k_i(\Gamma-a)^\top y\\
&=& \dot{T}^\top Mg v+k_i(\Gamma-a)^\top y\\
&=& y^\top\left( v+k_i(\Gamma-a)\right)\\
&=& -k_dy^\top y.
\end{eqnarray*}
In step 2 and 4 we use system dynamics \eqref{dyn} and controller dynamics respectively. In step 5 we used $\dot{g}+g\alpha=0$ given in Proposition \ref{stab_cond}. Finally in step 6 we have used the control strategy \eqref{Control_port}. Now one can infer that there exist an $\alpha>0$, such that 
\begin{eqnarray*}
\dot{V}_d&\leq& -\alpha \left(\left(\dot{T}_1-\dot{T}_3\right)^2+\left(\dot{T}_2-\dot{T}_4\right)^2+\left(\dot{T}_3-\dot{T}_4\right)^2\right.\\&&\left.+\dot{T}_1^2+\dot{T}_2^2\right).
\end{eqnarray*}
 $\dot{V}_d=0$ implies $\dot{T}_1$, $\dot{T}_2$, $\dot{T}_3$ and $\dot{T}_4$ are identically zero. Using this in \eqref{dyn}, we get $u_1$ and $u_2$ as constant. From \eqref{udot} we get $v=0$, substituting this in \eqref{Control_port} we get that $T_1=T_1^\ast$, and $T_2=T_2^\ast$. Finally, we conclude the proof by invoking LaSalle's invariance principle.
\end{proof}
{\em Simulation results:}
In order to illustrate the efficacy of the proposed approach an illustrate example of building thermal zone model is considered. The description of building zone model and the controller design are detailed in Section \ref{buildingzone}. The parameter values used for the simulation study is given in \cite{deng2010building}. The trajectories of zone temperatures for the two zone case is shown in Fig. \ref{fig:tres} and the effectiveness of controller is shown by zone temperatures reach their respective reference temperature values. The control inputs to the zones and the time evolution of port variables is shown in Fig. \ref{fig:control} and Fig. \ref{fig:ports}. Zone 2 needs higher control effort to reach reference temperature compared to zone 1 due to the higher difference in initial and reference values.
\begin{figure}[h!]
	\centering
	\includegraphics[width=0.9\linewidth]{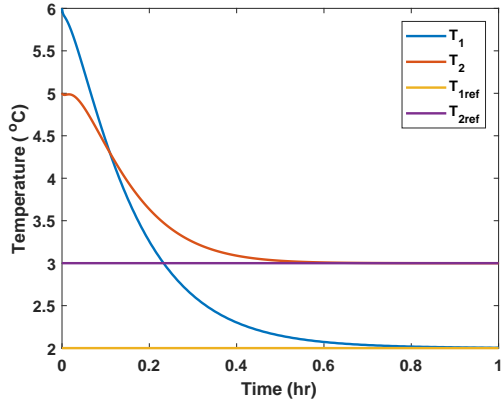}
	\caption{Trajectories of zone temperatures ($T_{\text{1ref}}=2.5$, $T_{\text{2ref}}=6$)}
	\label{fig:tres}
\end{figure}
\begin{figure}[h!]
	\centering
	\includegraphics[width=0.9\linewidth]{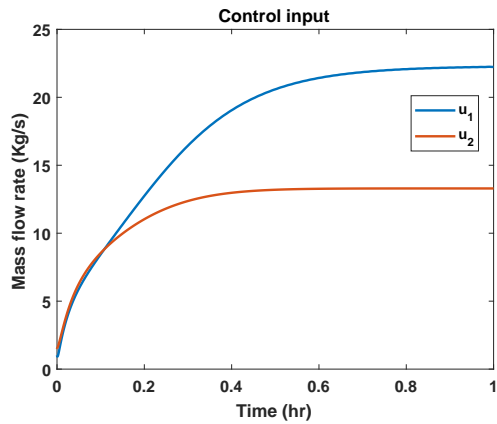}
	\caption{Time evolution of mass flow rate $u$.}
	\label{fig:control}
\end{figure}
\begin{figure}[h!]
	\centering
	\includegraphics[width=0.9\linewidth]{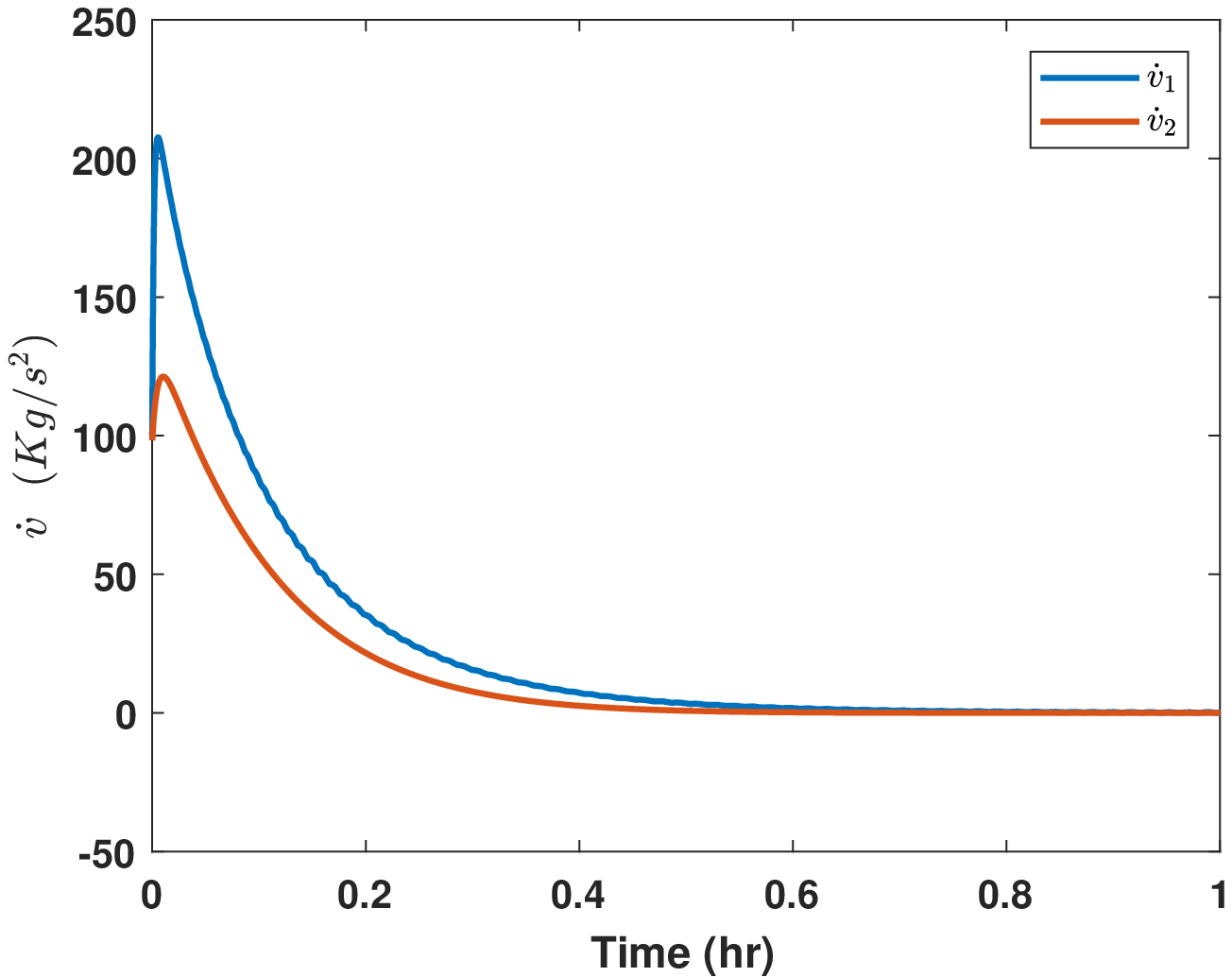}
	\caption{Time evolution of port variable $\dot{v}$ .}
	\label{fig:ports}
\end{figure}
	\section{Relations to differential and incremental passivity}\label{sec::diff_passvity}
	In this section we consider the prolonged system \cite{cortes2005characterization,crouch1987variational}, that is the original non-linear system together with its variational system.
	The prolonged system of \eqref{gen_sys} together with the variational version of input dynamics in equation \eqref{input_dyn} are
	\begin{eqnarray}\label{incre_extended_sys}
	    \dot{x}&=&f(x)+g(x)u\nonumber \\
	    \dot{\delta x} &=& \left(\dfrac{\partial f}{\partial x}+\dfrac{\partial g}{\partial x} u\right)\cdot \delta x+g(x)\delta u\\
	    \delta{u} &= &\alpha u+\beta +\delta v\nonumber
	\end{eqnarray}
	where $\delta x\in \mathbb{R}^n$, $\delta u\in \mathbb{R}^m$ denotes the variation in $x$ and $u$ respectively, 	$\alpha=-\left(g^\top g\right)^{-1}g^\top \dfrac{\partial g}{\partial x}\delta x$ and $\beta=-g^\top M \delta{x}$. Note that one can show $\dfrac{\partial g}{\partial x}\delta x+g \alpha =0$ using a similar procedure given in lemma \ref{prop::alpha}.
	\begin{proposition}
	The system of equations \eqref{incre_extended_sys} are passive with port variable $\delta y=g^\top M \delta x$ and $\delta v$. 
	\end{proposition}
	\begin{proof}
	Consider the following storage function,
	\begin{equation}\label{incre_storage}
	    V(x,\delta x) = \dfrac{1}{2}\delta x^\top M\delta x
	\end{equation}
	The time derivative of storage function \eqref{incre_storage} along the trajectories of \eqref{incre_extended_sys} is\\
	$\dfrac{d}{dt}V(x,\delta x)$
	\begin{eqnarray}\label{Vdot}
	    &=&\delta x^\top M \dot{\delta x}\nonumber\\
	    &=&\delta x^\top M \left(\left(\dfrac{\partial f}{\partial x}+\dfrac{\partial g}{\partial x} u\right)\cdot \delta x+g(x)\delta u\right)\nonumber\\
	    &=&\delta{x}^\top \left(M\dfrac{\partial f}{\partial x}+\dfrac{\partial f}{\partial x}^\top M\right)\delta x\nonumber\\
	    &&+ \left(\left(\dfrac{\partial g}{\partial x} \delta x+\delta x Mg \alpha \right)u+g(x)\left(\beta +\delta v\right)\right)\nonumber\\
	    &\leq & \delta x^\top M g \delta v
	    =\delta y^\top \delta v\nonumber
	\end{eqnarray}
	\end{proof}
		This approach shows that there are direct implications between dynamic feedback passivation and variational passivity.
	\begin{proposition}\label{Prop::incre_storage1}
	Consider system \eqref{incre_extended_sys}, with a smooth output $\bar y=h(x)\in \mathbb{R}^m$    and the input matrix $g(x)$ satisfies assumption A2. If there exist a  positive definite matrix $M$ satisfying assumption A1 i.e
	\begin{eqnarray}\label{incre_storage1}
	    M\dfrac{\partial f}{\partial x}+\dfrac{\partial f}{\partial x}^\top M &\leq& 0 ~~~\text{and}\\
	    Mg&=& \dfrac{\partial h}{\partial x}^\top .\nonumber
	\end{eqnarray}
	then the system \eqref{incre_extended_sys} is passive with port variable $\delta \bar y=g^\top M \delta x$ and $\delta v$. 
	\end{proposition}
	\begin{proof}
	From equation \eqref{Vdot}, the time derivative of storage function \eqref{incre_storage} along the trajectories of \eqref{incre_extended_sys} is\\
	\begin{eqnarray*}
	    \dfrac{d}{dt}V(x,\delta x) \leq  \delta x^\top M g \delta v
	    =\delta x^\top  \dfrac{\partial h}{\partial x}^\top \delta v
	    =\delta \bar{y}^\top \delta v
	\end{eqnarray*}
	\end{proof}
	%
	\begin{remark}
	In the above Proposition \ref{Prop::incre_storage1}, $Mg= \frac{\partial h}{\partial x}^\top$ denotes assumption A3, that is, $Mg(x)$ is integrable. Further, if we consider $g(x)=B$, and $h(x)=Cx$, where $B\in \mathbb{R}^{n\times m}$ and $C\in \mathbb{R}^{m\times n}$ are constant, then we recover the conditions presented from incremental passivity in \cite{pavlov2008incremental}.
	\end{remark}
	\section{Conclusion}
%

	 In this paper, Krasovskii's method of Lyapunov function is used for stability analysis and control for a class of nonlinear dynamical systems. The use of such Lyapunov functions has led to new passive maps which is used for controller design. The proposed approach is tested on a building zone model and controller is designed to maintain the desired setpoint temperature.   In Section \ref{sec::diff_passvity}, we have shown that the prolonged system together with the input dynamics satisfying the sufficient conditions leads to differential passivity. The sufficient conditions also relate to incremental passivity conditions as shown in remark 4.3. There is a natural connection between dynamic feedback passivation and variational passivity, which the authors would like to explore in future work.  
	\bibliographystyle{IEEEtran}
	\bibliography{refs}
\end{document}